\RequirePackage{snapshot}
\newif\ifproofs\proofsfalse\ifproofs\RequirePackage[displaymath,mathlines]{lineno}\fi
\RequirePackage{lineno}

\newif\ifsubsections
\documentclass[]{pcmi}

\usepackage[sc]{mathpazo}          % Palatino with smallcaps as text font
\usepackage{eulervm}               % Euler math
\usepackage[mathscr]{eucal}     % Euler math
\usepackage[scaled=0.86]{berasans} % Bera for san serif family
\usepackage[scaled=1]{inconsolata} % Inconsolata for fixed width
\usepackage[T1]{fontenc}
\usepackage{overpic}
\usepackage{float}
\usepackage[%
	protrusion=true,
	expansion=false,
	auto=false
	]{microtype}

\usepackage{xcolor}
\usepackage{graphicx}
\graphicspath{{./figures/}}

\long\def\commentout#1{}
\ifdraft
	\definecolor{linkred}{rgb}{0.7,0.2,0.2}
	\definecolor{linkblue}{rgb}{0,0.2,0.6}
\else
	\definecolor{linkred}{rgb}{0.0,0.0,0.0}
	\definecolor{linkblue}{rgb}{0,0.0,0.0}
\fi

\usepackage{tikz}
\usetikzlibrary{calc,arrows,shapes,positioning}
\usetikzlibrary{decorations.markings}
\tikzstyle arrowstyle=[scale=1]
\tikzstyle directed=[postaction={decorate,decoration={markings,
    mark=at position .65 with {\arrow[arrowstyle]{stealth}}}}]
\tikzstyle reverse directed=[postaction={decorate,decoration={markings,
    mark=at position .65 with {\arrowreversed[arrowstyle]{stealth};}}}]
\usepackage{float}
\usepackage[skip=6pt]{caption}

\PassOptionsToPackage{hyphens}{url}
\usepackage[
    setpagesize=false,
    pagebackref,
	plainpages=false,pdfpagelabels,pageanchor=true,
    pdfstartview={FitH 1000},
    bookmarksnumbered=false,
    linktoc=all,
    colorlinks=true,
    anchorcolor=black,
    menucolor=black,
    runcolor=black,
    filecolor=black,
    linkcolor=linkblue,%IM black if not in draft mode
	citecolor=linkblue,%IM black if not in draft mode
	urlcolor=linkred,%IM black if not in draft mode
]{hyperref}%IM
\usepackage[backrefs,msc-links,nobysame]{amsrefs}
\customizeamsrefs

\theoremstyle{plain}
\newtheorem{theorem}[equation]{Theorem}
\newtheorem{lemma}[equation]{Lemma}

\newtheorem{prop}[equation]{Proposition}

\theoremstyle{definition}

\newtheorem*{remark}{Remark}

\numberwithin{equation}{section}
\newcommand{\N}{\mathbb{N}}
\newcommand{\R}{\mathbb{R}}

\newcommand{\C}{\mathbb{C}}

\newenvironment{problem}[2][Problem]{\begin{trivlist}
		\item[\hskip \labelsep {\bfseries #1}\hskip \labelsep {\bfseries #2.}]}{\end{trivlist}}

\DeclareRobustCommand{\ddbar}{%
  {\mathord{\text{\lower0.05ex\hbox{$\mathchar'26$}{$\mkern-11mu d$}}}}%
}
\newcommand{\dts}{\ddbar s}

\newcommand{\beq}{\begin{equation}}
\newcommand{\eeq}{\end{equation}}
\newcommand{\beqs}{\begin{equation*}}
\newcommand{\eeqs}{\end{equation*}}

\begin{document}
%	\linenumbers
	\title{Long-time asymptotic behaviour for the fifth order Modified Korteweg-de Vries equation }
	\author{Fudong Wang,Wen-Xiu Ma}
	\address{Department of Mathematics and Statistics, University of South Florida, Tampa }
	\email{
		fudong@mail.usf.edu}
	\date{\today}
	
	\begin{abstract}%
	Following Deift-Zhou's nonlinear steepest descent method, the long-time behaviour for the Cauchy problem of the 5th modified Korteweg-de Vries equation is analysed. Based on the inverse scattering transform, the 5th order MKdV is transformed to an 2 by 2 oscillatory Riemann-Hilbert problem, then by manipulating the Cauchy operator and reducing the degree of phase function, the long-time asymptotic of the solution is given in terms of solutions of the parabolic cylinder equation.
	\end{abstract}
	
	\maketitle
	\tableofcontents
	\thispagestyle{empty}
	
		\section{Introduction}
	In this article we apply Deift-Zhou's non-linear steepest descent method to the 5th order MKdV equation:
	\beq\label{eq:5th mkdv}
	q_t=30q^4q_x-10q^2q_{xxx}-40qq_{xx}q_x-10q_x^3+q_{xxxxx}
	\eeq
	which belongs to the AKNS hierarchy~\cite{AKNS1974}. As is well known in the integrable system theory,by means of inverse scattering transform, for any Schwartz initial data, the solution to the equations of the AKNS hierarchy exists globally in time. However there is no way to solve them explicitly in general. In the past two decades, asymptotic method plays a crucial role in the integrable system theory.  The first systematic method to study the long-time asymptotic behaviour is due to Deift and Zhou\cite{deift_steepest_1993}. In Deift and Zhou's work, they direct consider a Riemann-Hilbert problem(RHP) and deform it to a model RHP which can be solve in terms of solutions of the Weber's parabolic cylinder equation. This can be consider as a nonlinear generalization of the classical method of steepest descent. Later on in 1996, Varzugin\cite{Varzugin1996} generated the classical method of stationary phase to solve the oscillatory RHP, and worked out the asymptotic expansions for the whole AKNS hierarchy with Schwartz initial data, where the error term is controlled by $O(t^{-3/4}\log(t))$. Many works are done for the KdV, NLS, mKdV, etc, which are all of order three or less. Recently, many studies about long time asymptotic form some 3 by 3 Riemann-Hilbert problem even 4 by 4 RHP are showing up, see for example\cite{MA2019,GENG2019151,Boutet2013} In some sense, the long-time asymptotic behaviour for the 5th order MKdV with Schwartz initial data was included in Varzugin's work implicitly. The purpose of writing this article is to study the long-time asymptotic of the 5th order MKdV explicitly.
	
	In our work, we will using the Deift-Zhou's method to study the long-time asymptotic behaviour of 5th MKdV which gives a better error terms ($O(t^{-1}\log(t))$) comparing to the method of stationary phase. The difficulties come from the high order of the phase function, which is a fifth order polynomial in our case. We will mainly follow Deift and Zhou's paper and do the necessary adjustments whenever involve the phase function. 
	
	The main theorem is as following:
	\begin{theorem}  Given $q(x,0)\in \mathcal{S}(\R)$ and its associated reflection coefficient $r(z)$,in the linear-like oscillation\cite{AS1977} region $-x=O(t)$,the long-time behaviour of the solution to the 5th MKdV, i.e. the leading term of the solution $q(x,t)$ to MKdV as $t\rightarrow \infty$, can be written as following:
		\begin{equation}
		\label{main result}
			q(x,t)=-2(\frac{v}{640tz_0^3})^{1/2}\cos{\left(-128tz_0^5+v\log{(2560tz_0^5)}+\phi(z_0)\right)}+O(\frac{\log(t)}{t})
		\end{equation}
		where
		\begin{equation}
		\phi(z_0)=\frac{5\pi}{4}-\arg(\bar{r}(-z_0))-\arg{(\Gamma(-\nu i))}+\frac{1}{\pi}\int_{-z_0}^{z_0}\log\frac{1-|r(s)|^2}{1-|r(-z_0)|^2}\frac{ds}{s+z_0}
		\end{equation}
		and $z_0=(|\frac{x}{80t}|)^{1/4}$.

	\end{theorem}
	
	The outline of the article is as following: In section 2 we simply formula the inverse scattering transform for the 5th MKdV and its connection to a oscillatory RHP. In section 3 we will introduce the solution method to RHP due to Beals and Coifman\cite{Beal1984}, which connects a singular integral equation with RHP. In section 4 a scalar RHP is introduce along with some estimates on the solution. This scalar RHP will be used to conjugate the matrix RHP and in preparing for contour deformation. Also some of the estimates will be used when we reduce the original RHP to a model RHP. In section 5, we will give the fundamental decomposition lemma in order to decompose a Schwartz function three parts, a non-analytic small term, a analytic term and a rational function. Contour deformation and truncation will be based on this lemma. In section 6 and 7 we will perform contour deformation($\Sigma^\sharp$) and truncation ($\Sigma'$) of the RHP, and give the error terms. In section 8, we will reduce the phase function (5th order) to a second order phase. Then we will separate the contributions of the two crosses $\Sigma_{A'}$ and $\Sigma_{B'}$. In section 9, reduce the RHP on $\Sigma_{A'(B')}$ to a model RHP and solve it in terms of parabolic cylinder equation.
	
	\section{IST and RHP}
	\subsection{Inverse scattering problem formulaism}
	In this section, we will formula the scattering and inverse scattering for initial value $q(x,t=0)\in \mathcal{S}(\R)$. First we consider the direct scattering problem and set $t=0 $:
	\begin{equation}
	\label{directscattering}
	\psi_x(x;z)=\left(iz\sigma_3+\begin{pmatrix}
	0&q(x,t)\\
	\overline{q(x,t)}&0\\
	\end{pmatrix}\right)\psi(x;z)
	\equiv U\psi\end{equation}
	where $\sigma_3=\begin{pmatrix}
	1&0\\
	0&-1\\
	\end{pmatrix}$.
	Then following the standard scattering method\cite{Beal1984,AS1977}, set $\mu=\psi e^{-ixz\sigma_3}$, and rewrite equation\eqref{directscattering} as following:
	\begin{equation}\label{eq:mu equation}
	\partial_x\mu=iz[\sigma_3,\mu]+Q\mu,\quad Q=\begin{pmatrix}
	0&q(x,t)\\
	\overline{q(x,t)}&0\\
	\end{pmatrix}
	\end{equation}
	In order to analysis the propeties of $\mu$, consider the following two integral equations($z\in \R$):
	\begin{equation}
	\mu_{\pm}(x;z)=I+\int_{\pm \infty}^x e^{i(x-y)\text{ad }\sigma_3}Q(y)\mu_{\pm}(y;z)dy
	\end{equation}
	By method of The Neumann series for the Volterra equations, we will see that these equations have unique bounded continuous solutions for $x,z\in \R$ provided that $q\in \mathcal{S}$.
	
	From ODE theory, any two solutions of equation\eqref{directscattering} are connected by a matrix independent of $x$, i.e. $\psi_+=\psi_-S(t;z)$, where $\psi_{\pm}=\mu_{\pm}e^{ixz\sigma_3}$. Since $\mu_{\pm}$ are normalized to identity matrix as $x\rightarrow \pm \infty$, it is easy to check that the scattering matrix $S(z)$ has determinant 1. And by symmtry of potential matrix, we have $S(z)=\begin{pmatrix}
	a&\bar{b}\\
	b&\bar{a}
	\end{pmatrix}$, where $|a|^2-|b|^2=1$. By analysis the Wronskians of $\psi_{\pm}$, we will have
	\begin{equation}
	\begin{split}
	a(z)&=1-\int_{\R}q(y)\mu^+_{21}(y;z)dy\\
	b(z)&=-\int_{\R}e^{2iyz}q(y)\mu_{11}^-(y;z)dy
	\end{split}
	\end{equation}
	Then $a$ can be analytic continuated to the upper half plane $\C_+$. And $a$ is continuous in $\overline{\C_+}$ and non-vanishing there. Moreover $a(\infty)=1$. Define the reflection coefficient $r:=-\bar{b}/\bar{a}$. In the present paper, we consider solitonless region, i.e. we assume $b\neq 0$. Since $|a|^2=1+|b|^2$, $|a|\geq 1$ and $|r|=1-|a|^{-2}<1$, therefore, $\|r\|_{L^\infty(\R)}<1$. And it is well-known that the direct scattering can be consider a bijective map $\mathcal{R}$ from the initial value $q(x,t=0)$ to the reflection coefficient $r(z)$. More over from \eqref{eq:mu equation}, and providing that $q\in \R$, $\bar{\mu}(-\bar{z})$ also satisfies the equation \eqref{eq:mu equation}, then by uniqueness, we obtaion:
	\begin{equation}
	\label{eq:symmetry of scattering }
	\bar{S}(-\bar{z})=S(z)
	\end{equation}
	and
	\begin{equation}
	\label{eq:symmetry of reflection}
	\bar{r}(-\bar{z})=r(z)
	\end{equation}
%	\begin{theorem}
%		$\mathcal{R} $ maps $\mathcal{S}$ into $\mathcal{S}$ bijectively.
%	\end{theorem}
%	\begin{corollary}
%		If $q(x,t=0)\in \mathcal{S}$, then $r(z)\in \mathcal{S}$ and $|r(z)|<1$
%	\end{corollary}

	Also from the analysis of the Neumann series of these two Volterra equations, the first column of $\mu_+$ and the second column of $\mu_-$, denoted by $\mu_{+1},\mu_{-2}$ respectively, can be extended to $\C_+$ analytically. Similarly, one sees that the first column of $\mu_-$ and the second column of $\mu_+$, denoted by $\mu_{-1},\mu_{+2}$ respectively, can be extended to $\C_-$ analytically By defining a new matrix
	\begin{equation}
	\label{rhp0matrix}
	m(x;z)=\begin{cases}
	(\frac{\mu_{+1}}{a(z)},\mu_{-2}),\quad z\in \C_+\\
	(\mu_{+2},\frac{\mu_{-1}}{b(z)}),\quad z\in \C_-\\
	\end{cases}
	\end{equation}
	Also denote $m_{\pm}$ as the boundary value of $m$ from $\C_+$ and $\C_-$ respectively
	Then by uniqueness of solutions of equation \eqref{directscattering}, we can that since $\psi_+=\psi_-S$,then $\mu_{+}=\mu_{-}e^{ix \text{ ad} \sigma_3}S$.Then by reordering the columns of $\mu_{\pm}$, one can see that there is a matrix $v$ such that $m_+=m_-v$. Direct calculation shows $v(z)=\begin{pmatrix}
	1-|r|^2& r\\
	-\overline{r}& 1\\
	\end{pmatrix}$. One thing worth mention here is that during the calculation, we will see naturally the factorization of matrix $v=\begin{pmatrix}
	1 & r\\
	0&1\\
	\end{pmatrix}\begin{pmatrix}
	1&0\\
	-\bar{r}& 1\\
	\end{pmatrix}$, which will be used in later sections.
	
	We summary the above direct scattering problem as the following Riemann-Hilbert problem:
	
	\begin{problem}{0} Given a jump condition $v(x,t=0;z)=e^{ix\text{ad}\sigma_3}\begin{pmatrix}
		1-|r|^2& r\\
		-\overline{r}& 1\\
		\end{pmatrix}=e^{ix\text{ad}\sigma_3}v(z),r\in \mathcal{S}$, on the real line $\R$. We are seeking for a $2\times 2$ matrix valued function $m(x;z)$ satisfies the following conditions:
		\begin{equation}
		\begin{cases}
		m(x;z)\text{ is analytic off the real line and continuous to the boundary}\\
		m_+=m_-v(x,t=0;z)\text{ on }\R\\
		m=I+\frac{m_1}{z}+o(z)\text{ as } z\rightarrow \infty
		\end{cases}
		\end{equation}
		\subsection{Time evolution and inverse scattering problem}
		In this section we will briefly discuss the time evolution and the inverse scattering problem and formulate them as a time evolution Riemann-Hilbert Problem. Since 5th MKdV is in the AKNS-hierarchy, the Lax pair of time evolution part corresponding to the direct scattering problem\eqref{directscattering} can be calculated by using some symbolic computation system. Here following Ma's scheme\cite{MA2013So3}, the stationary zero curvature equation $W_x=[U,W]$,where
		\[W=\sum_{i\geq 0}W_{0,i}\lambda^{-i},W_{0,i}=\begin{pmatrix}
		a_i& b_i\\
		c_i&-a_i\\
		\end{pmatrix}\]
		leads to the following recursion relation:
		\begin{equation}
		\begin{cases}
		b_{i+1}=\frac{1}{2I}b_{i,x}-Iqa_i,\\
		c_{i+1}=-\frac{1}{2I}-I\bar{q}a_i\\
		a_{i+1,x}=qc_{i+1}-\bar{q}b_{i+1}
		\end{cases}
		\end{equation}
		upon taking the initial values
		\begin{equation}
		a_0=16I,b_0=c_0=0
		\end{equation}
		also impose the conditions of the integration for the third recursion relation:
		\begin{equation}
		a_i|_{q=0}=b_i|_{q=0}=c_i|_{q=0}=0,\forall i \geq 1
		\end{equation}
		Now let 
		\begin{equation}
		V^{[m]}=(\lambda^mW)_+
		\end{equation}
		where $(\cdot)_+$ means the principle part of the Laurent expansion. Then the time-evolution problem is followed by 
		\begin{equation}
		\Psi_t=V^{[m]}\Psi,
		\end{equation}
		and the zero curvature equation
		\begin{equation}
		U_t-V^{[m]}_x+[U,V^{[m]}]=0
		\end{equation}
		leads to the equivalent non-linear integrable PDEs.
		For $m=2,3$, we will obtain the NLS equation and the MKdV equation respectively. In current paper, let $m=5$, we obtain the time-evolution part for the 5th MKdV equation, which reads 
		\begin{equation}
		\psi_t=(16Iz^5\sigma_3+V_0(q,\bar{q},z))\psi\equiv V\psi
		\end{equation}
	where 
	\begin{equation}
	\begin{split}
	V_0&=16qz^4\sigma_1+z^3(-8Iq^2\sigma_3+8Iq_x\sigma_1\sigma_3)+z^2(8q^3-4q_{xxx})\sigma_1\\
	&+z(6Iq^4-4Iq_{xx}q)+2Iq^2_{x}\sigma_3+12Iq^2q_x-2Iq_{xxx}\sigma_1\sigma_3)\\
	&+(6q^5-10q^2q_{xx}-10qq^2_x+u_{xxxx})\sigma_1
	\end{split}
	\end{equation}
	provided that $q=\bar{q}.$
	
	Then the time evolution of the reflection coefficient is given by \begin{equation}
		r(t)=r(t;z)=e^{-16itz^5}r(z)
		\end{equation}
		Now we formulate are the time evolution Riemann Hilbert problem as following:
	\end{problem}
	\begin{problem}{RHP1}
		Given $r(z)\in \mathcal{S}(\R)$,$v(x,t;z)=e^{(-16iz^5t+ixz)\text{ ad}\sigma_3}v(z)=e^{-it\theta(z;z_0)\text{ ad}\sigma_3}v(z)$. We are seeking to a $2\times 2$ matrix-valued function satisfying
		
		\begin{equation}
		\begin{cases}
		m(x,t;z)\text{ is analytic off $\R$ and continuous to $\R$}\\
		m_+(x,t;z)=m_-(x,t;z)v(x,t;z)\text{ on }\R\\
		m(x,t;z)=I+\frac{m_1(x,t)}{z}+o(z)\text{ as } z\rightarrow \infty
		\end{cases}
		\label{RHP1}
		\end{equation}
		where the phase function $\theta(z;z_0)=16z^5-80z^4_0z,z^4_0=-\frac{x}{80t}$, $z_0>0$, in this paper we only consider the region $x<0,t>0$ and $-x=O(t),$ as $t\rightarrow \infty$.
	\end{problem}
	
	Since $m$ by the definition of \eqref{rhp0matrix} also satisfies the equation \eqref{eq:mu equation}, then let $z\rightarrow \infty$ in both side,  we obtain that the solution to the Cauchy problem of the 5th order MKdV is:
	\begin{equation}
	\label{recoveringsolution}
	q(x,t)=-\lim_{z\rightarrow \infty}iz[\sigma_3,m]_{12}=-2i[m_1(x,t)]_{12}
	\end{equation}
	and the analysis of the long-time behaviour of solutions is reduced to the asymptotic analysis of RHP1.

	\section{Solution method of RHP by matrix factorization}
	In this section, we recall so-called Beal-Coifman method in order to solve the matrix RHP by factoring a matrix into two triangle matrices. First define the Cauchy Operator $C_{\pm}$, given a function $f\in L^2(\R)$,
	$$C_{\pm}f(z)=\lim_{\epsilon\downarrow 0}\int_\R\frac{f(s)}{s-(z\pm i\epsilon)}\frac{ds}{2\pi i}.$$ It is well known that those operator are bounded from $L^2$ to $L^2$. Also it worth note the property that $C_+-C_-=1$. Now consider a RHP with jump $$v=v_-^{-1}v_+=(1-w_-)^{-1}(1+w_+)$$ on some contour on $\R$, seek a function $m$ which is analytic in the upper half plane $\C_+$ and in the lower half plane $\C_-$ and continuous to the boundary from $\C_+$ or $\C_-$ respectively, denoted by $m_{\pm}$. On the boundary $m_+=m_-v$. The method of Beals and Coifman says that if $\mu$ solve the following singular integral equation:
	\begin{equation}
	\mu=1+C_w\mu
	\end{equation}
	where $$C_w(f):=C_-(fw_+)+C_+(fw_-).$$ Then the solution to the RHP is then given by
	\begin{equation}
	\label{RHPsol}
	m(z)=1+\int_{\R}\frac{\mu(s)(w_-(s)+w_+(s))}{s-z}\dts,\quad \dts=\frac{ds}{2\pi i}
	\end{equation}
	The existence of the RHP now transformed to existence of the singular integral equation, i.e. the invertibility of operator $I-C_w$. Also by the Fredholm theorem, the existence guarantees the uniqueness. It's easy to show $C_w$ is a bounded operator in $L^2$ provided that $w_{\pm}\in L^\infty$. One sufficient condition for $I-C_w$ to be invertible is given by 
	\begin{equation*}
	\|w_+\|_{L^\infty}+\|w_-\|_{L^\infty}<1
	\end{equation*} 
	In the following sections, we will factorize the jump matrix such that the $L^2$ norm corresponding $C_w$ operator will less than 1 for sufficient large $t$.
	
	Now suppose the RHP \eqref{RHP1} has a solution, combining \eqref{recoveringsolution} and \eqref{RHPsol},  the potential can be recovered by
	\begin{equation}
	q(x,t)=\left[\int_{\R}\mu(s)(w_-(s)+w_+(s))\frac{ds}{\pi}\right]_{12}
	\end{equation}
	
	\section{A scalar RHP}
In this section we will consider the following scalar $RHP$, given $r\in \mathcal{S},|r|<1$,seeking analytic function $\delta$, such that
\begin{equation}\label{eq:scalar rhp}
\begin{cases}
\delta_+(z)&=\delta_-(z)[\chi_{D_-}(1-|r|^2)+\chi_{D_+}]\\
\delta(\infty)&=1
\end{cases}
\end{equation} 
where $D_-=\{z:\theta(z)<0\}$ , $D_+=\{z:\theta(z)>0\}$ and $\chi$ is the characteristic function. Then the solution to \eqref{eq:scalar rhp} ,bu the Plemelj's formula, is  
\beq\label{eq:1}
\delta(z)=e^{\int_{-z_0}^{z_0}\frac{\log(1-|r(s)|^2)}{s-z}\dts}
\eeq
Also since $\log(1-|r(s)|^2)$ is Lipschitz continuous, by the Plemelj-Privalov theorem, $\delta(z)$ is also Lipschitz continuous on $\R$. More explicitly, set $\chi(z)=\int_{-z_0}^{z_0}\frac{\log(1-|r(s)|^2)}{\log(1-|r(-z_0)|^2)}\frac{\dts}{s-z}$ and $\nu=-(2\pi)^{-1}\log(1-|r(-z_0)|^2)$, we have the following formula near for $z\in \C\backslash \R$:
\beq
\label{eq:2}
\begin{split}
	\log(\delta(z))&=\int_{-z_0}^{z_0}\frac{\log(1-|r(s)|^2)}{s-z}\dts\\
	&=\int_{-z_0}^{z_0}\frac{\log(1-|r(s)|^2)-\log(1-|r(-z_0)|^2)}{s-z}+\frac{\log(1-|r(-z_0)|^2)}{s-z}\dts\\
	&=\chi(z)+\int_{-z_0}^{z_0}\frac{\log(1-|r(-z_0)|^2)}{s-z}\dts\\
	&=\chi(z)+i\nu
\end{split}
\eeq
with properly choosing cuts such that : \beq
|\arg(z\pm z_0)|<\pi.
\eeq
Also, given $q(x,t)$ is real function, then $\bar{\mu}(-\bar{z})$ solves the equation \eqref{eq:mu equation}, which further implies that scattering matrix $s(z)=\bar{s}(-\bar{z})$ thus for the reflection coefficient, we have:
\beq\label{3}
r(z)=\bar{r}(-\bar{z})
\eeq
By uniqueness of the scalar RHP, \eqref{eq:2} and \eqref{3}, we obtain
\[\delta(z)=\overline{\delta(-\bar{z})}=(\overline{\delta(\bar{z})})^{-1}\]
and for real $z$,
\begin{gather}
|\delta_+(z)\delta_-(z)|=1,\label{scalarRHPidentity2.0}\\
\quad |\delta_{\pm}(z)|=1,\quad if\quad z\in D_+,\label{scalarRHPidentity2.1}\\
|\delta_+(z)|=|\delta_-^{-1}(z)|=(1-|r(z)|^2)^{1/2} \quad for\quad z\in D_-,\label{scalarRHPidentity3}
\end{gather}

Hence by the maximum principle, $|\delta(z)|^{\pm 1}<\infty$ for all $z$.

Now let we conjugate the RHP \eqref{RHP1} to the following RHP:
\begin{problem}{RHP1'}
	\begin{equation}
	\label{RHP1'}
	\begin{cases}
	m_+\delta_+^{-\sigma_3}=m_-\delta_-^{-\sigma_3}\delta_-^{\sigma_3}v(x,t;z)\delta_+^{-\sigma_3}\\
	m\delta^{-\sigma_3}(\infty)=I
	\end{cases}
	\end{equation}
\end{problem}
\begin{remark}
	The normalization condition can be verified from the proposition since $r\in \mathcal{S}(\R)$,$\delta\rightarrow 1$ as $z\rightarrow \infty$.
\end{remark}

	\section{Decomposition of the Schwartz function}
	In this section, consider a fundamental decomposition of any Schwartz function in the spirit of method of stationary phase method, i.e decomposition a Schwartz function on the intervals where the phase function is monotonic there. This fundamental decomposition will be applied to decompose the matrix RHP.
	\begin{lemma}Suppose $\rho \in \mathcal{S}(\R)$,and given a phase function $\theta(z)=16z^5-80z_0^4z$,where $z_0$ the positive stationary point of $\theta$. Then there exists a decomposition of $\rho=h_1+h_2+R$,such that for $\epsilon>0$ and $\alpha\in (0,\pi/4]$, 
		\begin{equation}
		\begin{split}
		|e^{-2it\theta(z)}h_1(z)|&\leq ct^{-k},\quad z(u)=z_0+uz_0e^{\pi i},u\in [0,1]\\
		|e^{-2it\theta(z)}h_2(z)|&\leq ct^{-k},\quad z(u)=z_0+uz_0e^{i(\pi-\alpha)},u\in [0,1/\cos{\alpha}]\\
		|e^{-2it\theta(z)}R(z)|&\leq Ce^{-4z_0^5\epsilon^2t},\quad z(u)=z_0+uz_0e^{i(\pi-\alpha)},u\in [\epsilon,1/\cos{\alpha}]\\
		\end{split}
		\label{decopestimate}
		\end{equation}
	\end{lemma}
	\begin{proof}
		Since $\rho\in \mathcal{S}$,consider the Taylor truncate $R(z)=\sum_{j=0}^nc_j(z-z_0)^j$ with $n>1$ , and define $h=\rho-R$, then we have $h=O(|z-z_0|^{n+1})$. Set $a(z)=(z-z_0)^q,n>q\geq1,q\in\N$. Define a new function $f(\theta)=\{h/a\}(z(\theta))H(-|\theta|+64z_0^5)$. It is well-defined since $\theta(z)$ is monotonic in $[-z_0,z_0]$ hence is invertible. Note that by chain rule, we have
		\begin{equation}
		\frac{df}{d\theta}=\frac{df}{dz}(\theta')^{-1}
		\end{equation}
		Each time this process will reduce the degree of $z-z_0$ by 2. Also we have $d\theta=80(z^2+z_0^2)(z+z_0)(z-z_0)dz$. So near $z_0$,for any $0\leq j\leq \frac{n+1-q}{2} $, $\frac{d^jf}{d\theta^j}=O((z-z_0)^{n+1-q-2j+1})\in L^2(\R)$, then by the Plancherel's theorem, we have $(1+s^2)^{j/2}|\hat{f}|\in L^2(\R)$.
		
		Now consider $h(z)=a(z)\int_{\R}\hat{f}(s)e^{-is\theta}d\theta=a(z)\int_{t}^{\infty}\hat{f}(s)e^{is\theta}ds+a(z)\int_{-\infty}^{t}\hat{f}(s)e^{is\theta}ds$, and set $h_1=a(z)\int_{t}^{\infty}\hat{f}(s)e^{is\theta}ds,h_2=a(z)\int_{-\infty}^{t}\hat{f}(s)e^{is\theta}ds$. Then on the real line,
		\begin{equation}
		\begin{split}
		|e^{-2it\theta}h_1|&\leq c\int_t^\infty |\hat{f}|ds\\
		&\leq c\|(1+s^2)^{-p}\|_{L^2([t,\infty))}\|(1+s^2)^{p}|\hat{f}|\|_{L^2([t,\infty))}\\
		&\leq ct^{-p}
		\end{split}
		\end{equation}
		
		On the second segment of \eqref{decopestimate},we have 
		\begin{equation}
		\begin{split}
		|e^{-2it\theta}h_2|&c\leq(z_0u)^qe^{-t\Re i\theta(z)}\int_{-\infty}^te^{(s-t)\Re i\theta(s)}|\hat{f}(s)|ds\\
		&\leq cz_0^qu^qe^{-t\Re i\theta(z)}\|(1+s^2)^{-1}\|_{L^2(-\infty,t)}\|(1+s^2)|\hat{f}(s)|\|_{L^2(-\infty,t)}\\
		&\leq cu^qe^{-t\Re i\theta(z)}
		\end{split}
		\end{equation}
		In fact,consider the identiy
		\begin{equation}
		\theta(z)=-64z_0^5+160z_0^5(z-z_0)^2\left(1+\frac{z-z_0}{z_0}+\frac{(z-z_0)^2}{2z_0^2}+\frac{(z-z_0)^3}{10z_0^3}\right)
		\end{equation} 
		Thus on the ray, $z=z_0+uz_0e^{(\pi-\alpha)i},u\leq 1/\cos{\alpha}$,noting that $\alpha$ is fixed and sufficient small, we have
		\begin{equation}
		\begin{split}
		\Re{(i\theta)}&=16z_0^5u^2\left(10\sin{(2\alpha)}-10u\sin{(3\alpha)}+5u^2\sin{(4\alpha)}-u^3\sin{(5\alpha)}\right)\\
		&\geq 16u^2z_0^5\left(10\sin{(2\alpha)}-10\sin{(3\alpha)}/\cos{(\alpha)}+5(\cos{(\alpha)})^{-2}\sin{(4\alpha)}-(\cos{(\alpha)})^{-3}\sin{(5\alpha)}\right)\\
		&=16c(\alpha)z_0^5u^2
		\end{split}
		\end{equation}
		Since for small $\alpha$, $\left(10\sin{(2\alpha)}-10u\sin{(3\alpha)}+5u^2\sin{(4\alpha)}-u^3\sin{(5\alpha)}\right)$ will be monotonically decreasing and it is also easy to check that this minimum is positive as long as $\alpha>0$. Now we have $\Re i\theta(z)\geq 16c(\alpha)z_0^5u^2$, and by taking the derivative of $u^qe^{-t\Re i\theta(z)}$, it is easy to see that $u^qe^{-t\Re i\theta(z)}$ is controlled by $ct^{-q/2}$. So we have
		\begin{equation}
		|e^{-2it\theta}h_2|\leq ct^{-q/2}
		\end{equation}
		Finally, we estimate $R$ on the third segment of \eqref{decopestimate}. Since $R$ is a polynomial, it is controlled by the exponential decay, that is to say 
		\begin{equation}
		|e^{-2it\theta}R(z)|\leq ce^{-t\Re i\theta}\leq ce^{-16c(\alpha)z_0^5\epsilon^2t}
		\end{equation}
		Since the function $\rho \in \mathcal{S}$, $R$ could be a polynomial of any order, so are the $p$ and $q$. This completes the proof.
	\end{proof}
	\begin{remark}
		If we replace the $a(z)$ by $a(z)/(z+i)^2$, then the first two bounds on \eqref{decopestimate} should be $\frac{ct^{-k}}{1+|z|^2}$
	\end{remark}
	For the part of $(z_0,\infty)$ the decomposition is slightly different by replacing polynomial $R$ by a rational $R$ which decays at infinity. In the proof we need consider $\rho_0$ the Taylor truncate of $(z-i)^{10}\rho$ and set $R=\rho_0/(z-i)^{10}$ and $h=\rho-R$, then by the same harmonic analysis technique we obtain the following lemma:
	\begin{lemma}
		Suppose $\rho \in \mathcal{S}(\R)$,and given a phase function $\theta(z)=16z^5-80z_0^4z$,where $z_0$ the positive stationary point of $\theta$. Then there exists a decomposition of $\rho=h_1+h_2+R$,such that for $\epsilon>0$ 
		\begin{equation}
		\begin{split}
		|e^{-2it\theta(z)}h_1(z)|&\leq ct^{-k}/(1+|z|^2),\quad z(u)=z_0+uz_0e^{\pi i},u\in (-\infty,0)\\
		|e^{-2it\theta(z)}h_2(z)|&\leq ct^{-k}/(1+|z|^2),\quad z(u)=z_0+uz_0e^{(\pi-\alpha)i},u\in (-\infty,0)\\
		|e^{-2it\theta(z)}R(z)|&\leq ce^{-16c(\alpha)z_0^5\epsilon^2t},\quad z(u)=z_0+uz_0e^{(\pi-\alpha)i},u\in (-\infty,-\epsilon)\\
		\end{split}
		\label{decompest2}
		\end{equation}
	\end{lemma}
	And for phase function $-\theta(z)$, we have their counterparts. We summary all the estimates as the following theorem:
	\begin{theorem}
		\label{fundamentaldecomp}
		Let $\rho$ be a real-valued function in $\mathcal{S}(\R)$, given the phase function $\theta(z)=16z^5-80z_0^4z$,where $z_0$ is the only positive stationary point. Take $\epsilon\in (0,4z_0/5)$,and set 
		\begin{equation}
		\begin{split}
		L&= \{z:z=z_0+uz_0e^{(\pi-\alpha)i},u\in (-\infty,1/\cos{\alpha}]\}\\
		%	&\cup \{z:z=-z_0+uz_0e^{(\alpha)i},u\in (-\infty,1/\cos{\alpha}]\}\\
		L_\epsilon&=\{z:z=z_0+uz_0e^{(\pi-\alpha)i},u\in (\epsilon,1/\cos{\alpha}]\}\\
		%	&\cup \{z:z=-z_0+uz_0e^{(\alpha)i},u\in (\epsilon,1/\cos{\alpha}]\}\\
		\end{split}
		\end{equation} 
		there exist a decomposition $\rho=h_1+h_2+R$ satisfying the following estimates:
		\begin{eqnarray}
		\begin{split}
		|h_1(z)e^{-2it\theta(z)}|&\leq \frac{ct^{-k}}{1+z^2},\quad z\in \R,\forall k\in \N\\
		|h_2(z)e^{-2it\theta(z)}|&\leq \frac{ct^{-k}}{1+|z|^2},\quad z\in L,\forall k\in \N\\
		|R(z)e^{-2it\theta(z)}|&\leq ce^{-16c(\alpha)z_0^5\epsilon^2t},\quad z\in L_\epsilon
		\end{split}
		\end{eqnarray}
		Similarly, there exist a decomposition of $\bar{\rho}$ with respect to the $e^{2it\theta}$ on the conjugation of $L$ and $L_\epsilon$. And for stationary point $-z_0$ we have similar decomposition along with the estimates.
	\end{theorem}
	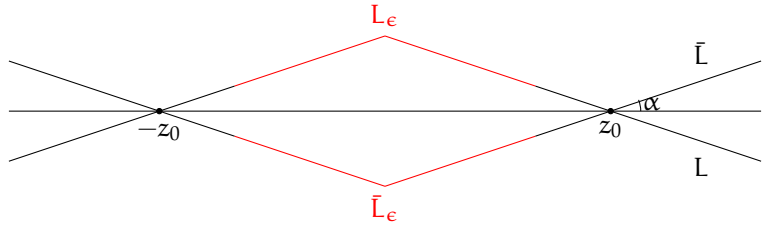
\begin{figure}[h]
		\centering
		\begin{tikzpicture}
		\draw (-5,0)--(5,0);
		\node [below] at (3,0) {$z_0$};
		\node [below] at (-3,0) {$-z_0$};
		\filldraw (3,0) circle (1pt);
		\filldraw (-3,0) circle (1pt);
		\draw [red] (0,1)--(2,1/3);
		\draw (2,1/3)--(5,-2/3);
		\draw [red] (0,1)--(-2,1/3);
		\draw (-2,1/3)--(-5,-2/3);
		\draw [red] (0,-1)--(2,-1/3);
		\draw (2,-1/3)--(5,2/3);
		\draw [red] (0,-1)--(-2,-1/3);
		\draw (-2,-1/3)--(-5,2/3);
		\draw (3.4,0) arc (0:22:0.4);
		\node [right] at (3.3,0.1) {$\tiny\alpha$};
		\node [above] at (4.2,1/2) {$\bar{L}$};
		\node [below] at (4.2,-1/2) {$L$};
		\node [above,red] at (0,1) {$L_\epsilon$};
		\node [below,red] at (0,-1) {$\bar{L}_\epsilon$};
		\end{tikzpicture}
		\caption{Contours of $L$ and $L_\epsilon$ and their conjugations for both $z_0$ and $-z_0$}\label{deformedContour}
	\end{figure}
	%\begin{proof}
	%	Most of the proof has been done from previous lemmas, we only need to show the case when $z\in L_{hor}$. In fact, let $z=u+i\epsilon,u\in [0,z_0-\epsilon]$, we have $\Re i\theta(z)\geq 80z_0^4\epsilon-26\epsilon^5>0$, then 
	%	\begin{equation*}
	%	\begin{split}
	%	|h_2(z)e^{-2it\theta(z)}|&\leq c(|z-z_0|)^qe^{-t\Re i\theta(z)}\int_{-\infty}^te^{(s-t)\Re i\theta(s)}|\hat{f}(s)|ds\\
	%	&\leq c(|(u-z_0)^2+\epsilon^2|)^qu^qe^{-t\Re i\theta(z)}\|(1+s^2)^{-1}\|_{L^2}\|(1+s^2)|\hat{f}(s)|\|_{L^2}\\
	%	&\leq cz_0^qe^{-t\Re i\theta(z)}\\
	%	&\leq cz_0^qe^{-t(80z_0^4\epsilon-26\epsilon^5)}
	%	\end{split}
	%	\end{equation*}
	%	Thus completes the proof.
	%\end{proof}
	%\begin{remark}
	%	We choose $\epsilon\leq 4z_0/5$ so that the $\Re i\theta$ can obtain its minimums and $e^{-2it\theta}$ will decay exponentially on the contour $L$.
	%\end{remark}
	
	\section{Contour Deformation of the RHP1}
	In this section, we deform the original RHP to the new contour $\Sigma=L\cup \bar{L}\cup \R$. For convenience we change the orientation for $|z|>z_0$, which is done by taking the inverse of the original jump matrix.
	
	Let $\Omega_j,j=1,2,...,8$ defined in Fig.\ref{deformRigen}, as well as the orientation of the contours.
	\begin{figure}[h]
		\centering
		\begin{tikzpicture}
		\draw (-5,0)--(5,0);
		\node [below] at (3,0) {$z_0$};
		\node [below] at (-3,0) {$-z_0$};
		\filldraw (3,0) circle (1pt);
		\filldraw (-3,0) circle (1pt);
		\draw [red,->] (0,1)--(2,1/3);
		\draw (2,1/3)--(5,-2/3);
		\draw [red] (0,1)--(-2,1/3);
		\draw [<-](-2,1/3)--(-5,-2/3);
		\draw [red,->] (0,-1)--(2,-1/3);
		\draw (2,-1/3)--(5,2/3);
		\draw [red] (0,-1)--(-2,-1/3);
		\draw [<-](-2,-1/3)--(-5,2/3);
		\draw [<-] (4,1/3)--(4.5,1/2);
		\draw [<-] (4,-1/3)--(4.5,-1/2);
		\draw [->] (-4,1/3)--(-4.5,1/2);
		\draw [->] (-4,-1/3)--(-4.5,-1/2);
		\draw (3.4,0) arc (0:22:0.4);
		\node [right] at (3.3,0.1) {$\tiny\alpha$};
		\node [above] at (4.2,1/2) {$\bar{L}$};
		\node [below] at (4.2,-1/2) {$L$};
		\node [above,red] at (0,1) {$L_\epsilon$};
		\node [below,red] at (0,-1) {$\bar{L}_\epsilon$};
		\node [above] at (5.1,0.1) {$\Omega_1$};
		\node [below] at (0,-0.3) {$\Omega_2$};
		\node [above] at (-5.1,0.1) {$\Omega_3$};
		\node [below] at (5.1,-0.1) {$\Omega_4$};
		\node [above] at (0,0.3) {$\Omega_5$};
		\node [below] at (-5.1,-0.1) {$\Omega_6$};
		\node [above] at (0.2,1.3) {$\Omega_7$};
		\node [below] at (0.2,-1.5) {$\Omega_8$};
		%orienations
		%\draw (0,1)->(2,1/3);
		\end{tikzpicture}
		\caption{Contours of $L$ and $L'$ and their conjugations for both $z_0$ and $-z_0$}\label{deformRigen}
	\end{figure}
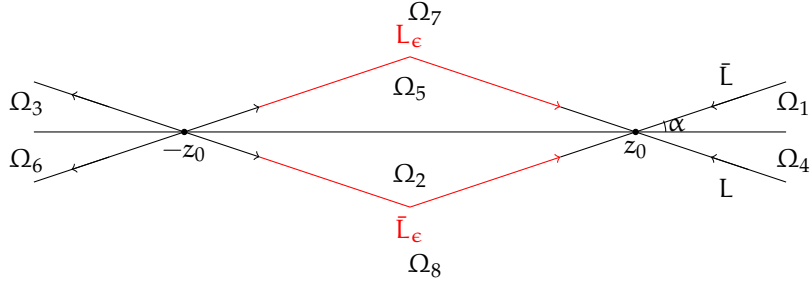
	And consider the factorization of jump matrix of RHP1', combining the decomposition lemma, let $\rho=-r(z)H(|z|>z_0)+\frac{r}{1-|r|^2}H(|z|<z_0)$, the jump matrix on \eqref{RHP1'} $\delta_-^{-\sigma_3}v\delta_+^{\sigma_3}$ can be rewritten as  \[ \begin{pmatrix}
	1&0\\-\bar{\rho}\delta_-^{-2} & 1
	\end{pmatrix}
	\begin{pmatrix}
	1&
	\rho\delta_+^{2}\\
	0&1
	\end{pmatrix}=:b_-^{-1}b_+.\] Due to the decomposition lemma $\rho=h_1+h_2+R$, denote \[b_-=b_-^ob_-^a=\begin{pmatrix}
	1&0\\\bar{h_1}\delta_-^{-2} & 1
	\end{pmatrix}\begin{pmatrix}
	1&0\\\overline{h_2+R}\delta_-^{-2} & 1
	\end{pmatrix}\]
	\[b_+=b_+^ob_+^a=\begin{pmatrix}
	1&
	h_1\delta_+^{2}\\
	0&1
	\end{pmatrix}\begin{pmatrix}
	1&
	(h_2+R)\delta_+^{2}\\
	0&1
	\end{pmatrix}\]
	and define
	\[w_{\pm}:=\pm(b_{\pm}-I)\]
	\[w=w_++w_-.\]
	
	It is easy to show that for fix $x,t$, we have $w_\pm,w\in L^2\cap L^1\cap L^\infty.$
	
	\begin{problem}{RHP2}
		Setting 
		\begin{equation}
		m^\sharp(z)=\begin{cases}
		m\delta^{-\sigma_3},&\quad z\in \Omega_7\cup \Omega_8\\
		m\delta^{-\sigma_3}(b_+^a)^{-1},&\quad z\in \Omega_4\cup\Omega_5\cup\Omega_6\\
		m\delta^{-\sigma_3}(b_-^a)^{-1},&\quad z\in \Omega_1\cup\Omega_1\cup\Omega_3\\
		\end{cases}
		\end{equation}
		\begin{equation}
		v^\sharp(z)=\begin{cases}
		(b_-^o)^{-1}b_+^o,&\quad z\in \R\\
		b_+^a,&\quad z\in L\\
		(b_-^a)^{-1},&\quad z\in \bar{L}\\
		\end{cases}
		\end{equation}
		
		\begin{equation}
		\label{RHP2}
		\text{RHP2}=
		\begin{cases}
		m^\sharp_+=m^\sharp_-e^{-it\theta\hat{\sigma_3}}v^\sharp\\
		m^\sharp(\infty)=I\\
		\end{cases}
		\end{equation}
	\end{problem}
	One thing we need to check is the normalization condition, $m^\sharp(\infty)=I$, which follows directly by using the estimates in the decomposition lemma and estimates for the scalar RHP.
	\section{Truncation of the contours}
	Following the analysis in Deift-Zhou's method,especially the restriction lemma (\cite{deift_steepest_1993} Lemma 2.56), since we have similar decomposition for $\rho$($\bar{\rho}$), we can easily estimate the errors generated from the truncating contours. For reader's convenience, we will list the lemmas which trivially follow from the decomposition lemma. Before that, we introduce some new notations first. 
	
	Set $w':\Sigma\rightarrow M(2,\C)$ supported on $\Sigma'=\Sigma\backslash (L_\epsilon\cup\bar{L}_\epsilon\cup \R)$ with contributions from $R(\bar{R})$ only. And denote the difference of $w^\sharp$ and $w'$ as $w^e$, see Fig.\ref{fig:truncate errors},  where $w^\sharp=w^\sharp_++w^\sharp_-$ and $w^\sharp_\pm=\pm(b^\sharp_\pm-I)$. In what follows, it takes two steps to reduce RHP-data $(w^\sharp,\Sigma)$ to $(w',\Sigma')$. The following estimate show the $L^2(dz)$ uniform boundedness for $w',w^\sharp,w^e$ while their $L^1$ and $L^\infty$ boundedness are directly from  the decomposition lemma.
	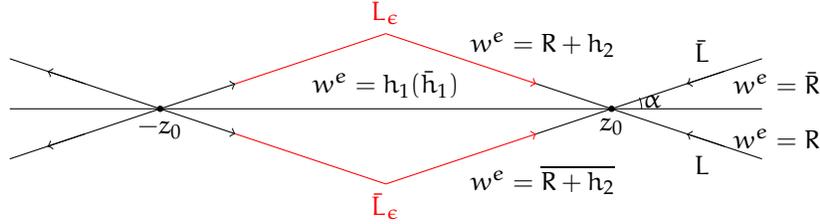
\begin{figure}[h]
		
		\centering
		\begin{tikzpicture}
		\draw (-5,0)--(5,0);
		\node [below] at (3,0) {$z_0$};
		\node [below] at (-3,0) {$-z_0$};
		\filldraw (3,0) circle (1pt);
		\filldraw (-3,0) circle (1pt);
		\draw [red,->] (0,1)--(2,1/3);
		\draw (2,1/3)--(5,-2/3);
		\draw [red] (0,1)--(-2,1/3);
		\draw [<-](-2,1/3)--(-5,-2/3);
		\draw [red,->] (0,-1)--(2,-1/3);
		\draw (2,-1/3)--(5,2/3);
		\draw [red] (0,-1)--(-2,-1/3);
		\draw [<-](-2,-1/3)--(-5,2/3);
		\draw [<-] (4,1/3)--(4.5,1/2);
		\draw [<-] (4,-1/3)--(4.5,-1/2);
		\draw [->] (-4,1/3)--(-4.5,1/2);
		\draw [->] (-4,-1/3)--(-4.5,-1/2);
		\draw (3.4,0) arc (0:22:0.4);
		\node [right]at (1,0.85)  {$w^e=R+h_2$};
		\node[right ] at (1,-0.95)  {$w^e=\overline{R+h_2}$};
		\node [right ]at (4.5,-0.4) {$w^e=R$};
		\node [right ]at (4.5,0.34) {$w^e=\bar{R}$};
		\node [right] at (3.3,0.1) {$\tiny\alpha$};
		\node [above] at (0,0) {$w^e=h_1(\bar{h}_1)$};
		\node [above] at (4.2,1/2) {$\bar{L}$};
		\node [below] at (4.2,-1/2) {$L$};
		\node [above,red] at (0,1) {$L_\epsilon$};
		\node [below,red] at (0,-1) {$\bar{L}_\epsilon$};
%		\node [above] at (5.1,0.1) {$\Omega_1$};
%		\node [below] at (0,-0.3) {$\Omega_2$};
%		\node [above] at (-5.1,0.1) {$\Omega_3$};
%		\node [below] at (5.1,-0.1) {$\Omega_4$};
%		\node [above] at (0,0.3) {$\Omega_5$};
%		\node [below] at (-5.1,-0.1) {$\Omega_6$};
%		\node [above] at (0.2,1.3) {$\Omega_7$};
%		\node [below] at (0.2,-1.5) {$\Omega_8$};
		%orienations
		%\draw (0,1)->(2,1/3);
		\end{tikzpicture}
		
		\caption{The errors of the truncation of contours.}
		\label{fig:truncate errors}
	\end{figure}
	\begin{lemma} \label{truncatelemma}
		$\|w^\sharp\|_{L^2(\Sigma,dz)}\leq Ct^{-1/4}$,\quad $\|w^e\|_{L^2(\Sigma,dz)}\leq Ct^{-k}$, \quad$\|w'\|_{L^2(\Sigma,dz)}\leq Ct^{-1/4}$
	\end{lemma}
	\begin{proof}
		The second estimate directly comes from the decomposition lemma. For the last estimate, consider $w'_+$ first, since $|R|\leq C(1+|z|^2)^{-1}$ for $z\in L$, and also we have
		\begin{equation*}
		\Re{(i\theta)}\geq 16c(\alpha)z_0^5u^2 
		\end{equation*} 
		Hence by the estimate for the scalar RHP, we have 
		\begin{equation}
		\begin{split}
		\|w'_+\|_{L^2}&\leq (\int_{L}|\delta^2(z)R(z)|^2e^{-2t\Re{(i\theta)}}|dz|)^{1/2}\\
		&\leq C(\int_{(-\infty,1/\cos{\alpha}]}e^{-32c(\alpha)z_0^5u^2t}du)^{1/2}\\
		&\leq Ct^{-1/4}
		\end{split}
		\end{equation}
		And we have similar estimate for $\bar{R}$, then by triangle inequality we have
		\begin{equation*}
		\|w'\|_{L^2}\leq \|w'_+\|_{L^2}+\|w'_-\|_{L^2}\leq Ct^{-1/4}
		\end{equation*}
		Then, the first estimate comes from the triangle inequality.
	\end{proof}
	The first reduction is to reduce $(w^\sharp,\Sigma)$ to $(w',\Sigma)$. It is essential to show the boundedness of $(1-C_{w'})^{-1}$ and $(1-C_{w^\sharp})^{-1}$ first. We will prove the following two propositions:
	\begin{prop}
		$(1-C_{w'})^{-1}$ is uniformly bounded form $L^{\infty}(\Sigma)+L^2(\Sigma)$ to $L^2(\Sigma)$ for $t$ sufficiently large.
	\end{prop}
	\begin{proof}
		It is equivalent to show that there exits $t_0$ such that for  $t>t_0$, $\|C_{w'}\|_{L^2}<1$, where $C_{w'}$ maps $L^\infty$ to $L^2$ since $w'(w'_{\pm})\in L^2$. In fact, taking $f\in L^\infty(\Sigma,dz)$ \begin{equation}
		\begin{split}
		\|C_{w'}f\|_{L^2}&\leq \|C_+(fw'_-)\|_{L^2}+\|C_-(fw'_+)\|_{L^2}\\
		&\leq (\|w'_+\|_{L^2}+\|w'_-\|_{L^2})\|f\|_{L^\infty}\\
		&\leq Ct^{-1/4}\|f\|_{L^\infty}
		\end{split}
		\end{equation}
		And by choosing $t_0=C^4+1$, we have $\|C_{w'}\|_{L^2}<1$ uniformly with respect to $z\in \Sigma$.
		 
		For $f\in L^2(\Sigma,dz)$,
		\begin{equation}
		\begin{split}
		\|C_{w'}f\|_{L^2}&\leq \|C_+(fw'_-)\|_{L^2}+\|C_-(fw'_+)\|_{L^2}\\
		&\leq \|w'_-\|_{L^\infty}\|C_+(f)\|_{L^2}+\|w'_+\|_{L^\infty}\|C_-(f)\|_{L^2}\\
		&\leq Ct^{-k}\|f\|_{L^2}
		\end{split}
		\end{equation} 
		and choosing $t_0=C^k+1$, we have $\|C_{w'}f\|_{L^2}<1$ uniformly with respect to $z$. 
		Then there exists a $t_0$, such that $C_{w'}$ is uniformly bounded from $L^2+L^\infty\rightarrow L^2$
	\end{proof}
	Similarly, we have the following proposition for $(w^\sharp,\Sigma)$:
	\begin{prop}
		$(1-C_{w^\sharp})^{-1}$ is uniformly bounded form $L^{\infty}(\Sigma)$ to $L^2(\Sigma)$ for $t$ sufficiently large.
	\end{prop}
	Then by the lemma \eqref{truncatelemma} and the resolvent identities, it is by direct computation to show the following estimate:
	\begin{lemma}\label{first truncate}
		\begin{equation}
		q(x,t)=2\left[-\int_{\Sigma}((1-C_{w'})^{-1}I)w'(s)\right]_{12}+O(t^{-k}),\forall k\in \N
		\end{equation}
	\end{lemma}
%\begin{proof}
%	********Using the resolvent identity. separate into five parts********
%	
%	
%	
%\end{proof}

	And recall the restriction lemma in Deift-Zhou's paper\cite{deift_steepest_1993},Lemma 2.56, we can restrict the Cauchy operator on $\Sigma$ to $\Sigma'$ without errors, i.e. $(1_{\Sigma'}-C_{w'}^{\Sigma'})^{-1}I=(1_{\Sigma}-C_{w'}^{\Sigma})^{-1}I$. And we finally have the following proposition:
	\begin{prop}
		\begin{equation}
		q(x,t)=2\left[-\int_{\Sigma'}((1-C_{w'})^{-1}I)w'(s)\right]_{12}+O(t^{-k}),\forall k\in \N
		\end{equation}
	\end{prop}
	The corresponding RHP reads:
	
	Set
	\begin{equation*}
	L'=L\backslash L_\epsilon
	\end{equation*}
	and
	\begin{equation}
	\Sigma'=L'\cup \overline{L'}
	\end{equation}
	Define the sectional analytic function $m'(z),z\notin \Sigma'$ as
	\begin{equation}
	m'(z)=I+\int_{\Sigma'}\frac{((1-C_{w'})^{-1}I)w'(s)}{s-z}\frac{ds}{2\pi i}
	\end{equation}
	On the boundary we have a new RHP:
	\begin{problem}{RHP3}
		\begin{equation}
		\begin{cases}
		m'_+=m'_-e^{-it\theta\hat{\sigma_3}}v'(z),z\in \Sigma'\\
		m'(\infty)=I
		\end{cases}
		\end{equation}
		where
		\begin{eqnarray}
		w'&=&w'_++w'_-,\\
		b'_\pm&=&I\pm w'_\pm\\
		v'&=&(b'_-)^{-1}b'_+
		\end{eqnarray}
		from the definition of $w'$ we have
		\begin{equation}
		\begin{cases}
		b'_+=\begin{pmatrix}
		1& R\delta_+^2\\
		0&1\\
		\end{pmatrix},b'_-=\begin{pmatrix}
		1& 0\\
		0&1\\
		\end{pmatrix},z\in L'\\
		b'_+=\begin{pmatrix}
		1& 0\\
		0&1\\
		\end{pmatrix},b'_-=\begin{pmatrix}
		1& 0\\
		\bar{R}\delta_-^{-2}&1\\
		\end{pmatrix},z\in \overline{L'}
		\end{cases}
		\end{equation}
	\end{problem}

	%\section{Principle of Localization}
	\section{Reducing the phase function and separating the contributions}
	In this section we will show how to reduce the order of phase function and then separate the contributions from different stationary phase points. First we introduce some new notations. Split $\Sigma'$ into disjoint union of two crosses $\Sigma'_A\cup \Sigma'_B$, see Figure \ref{splitedCrosses}. And decompose $w'=w'\chi_{\Sigma'_A}+w'\chi_{\Sigma'_B}=:w^{A'}+w^{B'}$.
	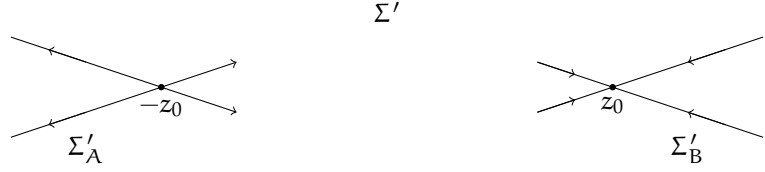
\begin{figure}[h]
		\centering
		\begin{tikzpicture}
		%	\draw (-5,0)--(5,0);
		\node [below] at (3,0) {$z_0$};
		\node [below] at (-3,0) {$-z_0$};
		\filldraw (3,0) circle (1pt);
		\filldraw (-3,0) circle (1pt);
		%\draw [red,->] (0,1)--(2,1/3);
		\draw (2,1/3)--(5,-2/3);
		%\draw [red] (0,1)--(-2,1/3);
		\draw [<-](-2,1/3)--(-5,-2/3);
		%	\draw [red,->] (0,-1)--(2,-1/3);
		\draw (2,-1/3)--(5,2/3);
		%	\draw [red] (0,-1)--(-2,-1/3);
		\draw [<-](-2,-1/3)--(-5,2/3);
		\draw [<-] (4,1/3)--(4.5,1/2);
		\draw [<-] (4,-1/3)--(4.5,-1/2);
		\draw [->] (-4,1/3)--(-4.5,1/2);
		\draw [->] (-4,-1/3)--(-4.5,-1/2);
		\draw [->] (2,1/3)--(2.5,1/6);
		\draw [->] (2,-1/3)--(2.5,-1/6);
		\node [below] at (4,-1/2){$\Sigma'_B$};
		\node [below] at (-4,-1/2){$\Sigma'_A$};
		\node at (0,1) {$\Sigma'$};
		\end{tikzpicture}
		\caption{Splitting $\Sigma'$ into $\Sigma'_A$ and $\Sigma'_B$}\label{splitedCrosses}
	\end{figure}
	
	Define the Cauchy operators $A',B'$ on $\Sigma'$ by
	\begin{equation}
	A':=C_{w^{A'}},\quad B':=C_{w^{B'}}
	\end{equation}
	Extend the contours $\Sigma'_A$ and $\Sigma'_B$ by assigning 0 to $\hat{w}^{A'}$,$\hat{w}^{B'}$ to the contours 
	\begin{eqnarray}
	\hat{\Sigma}_{A'}&=&\{z=-z_0+z_0ue^{\pm i \alpha}:u\in\R\}\\
	\hat{\Sigma}_{B'}&=&\{z=z_0+z_0ue^{\pm i \alpha}:u\in\R\}\\
	\end{eqnarray}
	The associated operators on $\hat{\Sigma}_{A'}$,$\hat{\Sigma}_{B'}$ are denoted by $\hat{A}',\hat{B}'$. And denote the shifted contours by $\Sigma_A,\Sigma_B$ , which are $\{z=z_0ue^{i\pm \alpha},u\in\R\}$ oriented as $\hat{\Sigma}_{A'}$,$\hat{\Sigma}_{B'}$ respectively. 
	
	Introduce the shift and scaling operators, set $a=\sqrt{640tz_0^3}$:
	\begin{eqnarray}
	N_A(f(z))&:=&f(z/a-z_0)\\
	N_B(f(z))&:=&f(z/a+z_0)\\
	\end{eqnarray}
	In the following analysis we will focus on the contour $\Sigma_B$ and give the similar results for the contour $\Sigma_A$ without proofs.
	
	In fact, we have
	\begin{equation}
	N_B(\delta 
	e^{-it\theta})(z)=\delta_B^0(z)\delta_B^1(z)
	\end{equation}
	where
	\begin{equation}
	\delta_B^0(z)=e^{\chi(z_0)}e^{-it\theta(z_0)}(2z_0a)^{-i\nu}
	\end{equation}
	and
	\begin{equation}
	\delta_B^1(z_0)=z^{i\nu}(\frac{2z_0}{z/a+2z_0})^{i\nu}e^{\chi(z/a+z_0)-\chi(z_0)}e^{-iz^2/4[1+z/(az_0)+z^2/(2a^2z_0^2)+z^3/(10z_0^3a^3)]}
	\end{equation}
	
	Since the only difference between our concern and Deift-Zhou 93's work\cite{deift_steepest_1993} is the phase function, and the estimates for the Cauchy operator after shifting and scaling are completely based on the phase. By conjuate the matrices $\hat{w}^{B}_\pm$ on the contour $\bar{L}_B=\{z=uz_0ae^{i\alpha},-\epsilon<u<\infty\}$ and $L_B=\{z=uz_0ae^{-i\alpha},-\epsilon<u<\infty\}$ and with identity jumps for the rest of the contour, i.e. $\Sigma_B\backslash (L_B\cup \bar{L}_B)$, as a result we have
	\begin{eqnarray}
	N^0_B(\hat{w}_-^{B'})&:=&(\delta_B^0)^{\hat{\sigma_3}}N_B(\hat{w}_-^{B'})=\begin{pmatrix}
	0 & 0 \\
	-\bar{R}(z/a+z_0)(\delta_B^1)^{-2} &0\\
	\end{pmatrix}\\
	N^0_B(\hat{w}_+^{B'})&:=&(\delta_B^0)^{\hat{\sigma_3}}N_B(\hat{w}_+^{B'})=\begin{pmatrix}
	0 & R(z/a+z_0)(\delta_B^1)^{2} \\
	0&0\\
	\end{pmatrix}\\
	\end{eqnarray}
	
	As $t\rightarrow \infty$,
	\begin{equation}
	\bar{R}(z/a+z_0)(\delta_B^1)^{-2}-\bar{R}({z_0}{\pm})z^{-2\nu i}e^{iz^2/2}\rightarrow 0
	\end{equation}
	where
	\begin{equation}
	R(z_0\pm)=\lim_{z\rightarrow z_0\pm}R(z)=\lim_{z\rightarrow z_0\pm}\rho(z)=\{+:\bar{r}(z_0);-:-\bar{r}(z_0)(1-|r(z_0)|^2)^{-1}\}
	\end{equation}

	More specifically, for the contour $\Sigma_{B}$, we have the following estimate for the rate of convergence:

	\begin{lemma}{Analogous to Lemma 3.35 in \cite{deift_steepest_1993}.}\label{lemma:phasereduction}
		Let $\gamma$ be a small positive number that $\gamma<1/2$ and $t_0$ be some large number. Then for $z\in \bar{L}_B$,and for $t>t_0$,
		\begin{equation}
		\begin{split}
		\left|\bar{R}(z/a+z_0)(\delta_B^1)^{-2}-\bar{R}({z_0}{\pm})z^{-2\nu i}e^{iz^2/2}\right|\\
		\leq C(z_0)e^{-\gamma \Im{z^2/2}}\left(t^{-1/2}+t^{-1/2}\log{(t)}\right)
		\end{split}
		\end{equation}
	\end{lemma}
	\begin{lemma} \label{lemma:ppt}
		Let $f(s)=\log{\frac{1-|r(s)|^2}{1-|r(z_0)|^2}}  $, $\chi{(z)}=\frac{1}{2\pi i}\int_{-z_0}^{z_0}\frac{f(s)}{s-z}ds$, where $r$ is the reflection coefficient which is in schwartz space. Then for $z\in L_0= \{z=ue^{i\alpha},|u|<1\}$,we have
		\begin{equation}
		\begin{split}
		|\chi{(z_0+z)}-\chi{(z_0)}|&\leq c|z||\log{|z|}|\\
		\end{split}
		\end{equation}
		
	\end{lemma}
	\begin{proof}
		Since $r$ is a Schwartz function, it is trivial to show that $f(s)$ is Lipschtz. And $f(z_0)=0$, so $|f(s)|=|f(s)-f(z_0)|\leq C|s-z_0|$, where $C$ is independent from $z_0,s$.
		
		Now write
		\begin{equation}
		\begin{split}
		|\chi(z+z_0)-\chi(z_0)|&\leq \frac{|z|}{2\pi}\int_{-z_0}^{z_0}\frac{|f(s)|ds}{|(s-z-z_0)(s-z_0)|}\\
		&\leq \frac{C|z|}{2\pi}\int_{-z_0}^{z_0}\frac{ds}{|s-z-z_0|}\\
		&\leq \frac{C|z|}{2\pi}\int_{-2z_0}^{0}\frac{ds}{|s-z|}
		\end{split}
		\end{equation}
		Since $|s-z|\geq 1/2(-s \sin{(\alpha)}+|z|\sin(\alpha)),\forall z\in L_0,s\in [-2z_0,0]$, we have
		\begin{equation}
		\begin{split}
		\int_{-2z_0}^0\frac{ds}{|s-z|}&\leq 2\int_{-2z_0}^0\frac{ds}{-s \sin{(\alpha)}+|z|\sin(\alpha) }\\
		&=\frac{2}{\sin \alpha}\int_{-2z_0}^0\frac{ds}{-s+|z|}\\
		&=\frac{2}{\sin \alpha}\int_{|z|}^{|z|+2z_0}\frac{ds}{s}\\
		&\leq \frac{2}{\sin \alpha}\log{(1+\frac{2z_0}{|z|})}\\
		&\leq C\left|\log|z|\right|
		\end{split}
		\end{equation}	
		Combining above analysis, lemma is proved.
	\end{proof}
	\begin{remark}
		In fact the above lemma is the direct conclusion from the Plemelj-Privalov theorem.
	\end{remark}
	\begin{proof}{The main lemma}
		
		%		5 steps.
		%Based on the uniformly boundedness!!!!
		Write
		\begin{equation}\label{phasereduceinequality}
		\begin{split}
		\bar{R}(z/a+z_0)&(\delta_B^1)^{-2}-\bar{R}({z_0}{\pm})z^{-2\nu i}e^{iz^2/2}\\
		=(e^{i\gamma z^2/2})&\left(e^{i\gamma z^2/2}\left[\bar{R}(z/a+z_0)(\frac{2z_0}{z/a+2z_0})^{-2i\nu}z^{-2i\nu}\right.\right.\\
		&\left.\left.e^{i(1-2\gamma)z^2/2\xi}e^{-2[\chi{(z/a+z_0)}-\chi{(z_0)}]}-\bar{R}({z_0}{\pm})z^{-2\nu i}e^{i(1-2\gamma)z^2/2}\right]\right)
		\end{split}
		\end{equation}
		where $\xi=1+(1-2\gamma)^{-1}(z/(az_0)+z^2/(2a^2z_0^2)+z^3/(10z_0^3a^3)).$
		Each terms in \eqref{phasereduceinequality} is uniformly bounded with respect to $x<0,t>0$. $|e^{i\gamma z^2/2}|=e^{-\gamma z_0^2u^2a^2 sin(\alpha)}$ is trivially bounded provided that $\alpha<\pi/2$, so is $|e^{i(1-2\gamma )z^2/2}|$. Applying the decomposition lemma, $|\bar{R}(z/a+z_0)|\leq c/(1+z_0^2)$. Also
		\begin{equation}
		\begin{split}
		\sup_{-\epsilon<u<\infty}&|(\frac{2z_0}{z/a+2z_0})^{-2\nu i}|\\
		&=\sup_{-\epsilon<u<\infty} e^{-2\nu arg(1+u/2e^{i\alpha})}\leq C
		\end{split}
		\end{equation}
		as $\arg(1+u/2e^{i\alpha})$ is positive when $u>0$ and $0<\nu<-1/(2\pi)\log(1-\eta^2)\leq \infty$ provided that $|r(z)|\leq \eta<1.$.
		The term $e^{i(1-2\gamma)z^2/2\xi}$ is bounded as
		\begin{equation}
		\begin{split}
		\Re&{i(1-2\gamma)z^2/2\xi}\\
		&=\Re i(1-2\gamma)1/2u^2a^2z_0^2e^{i2\alpha}(1+(1-2\gamma)^{-1}(ue^{i\alpha}+1/2u^2e^{2i\alpha}+1/10u^3e^{3i\alpha}))\\
		&=-1/2u^2a^2z_0^2(\sin(2\alpha)+u \sin(3\alpha)+1/2u^2 \sin(4\alpha)+1/10u^3 \sin(5\alpha))
		\end{split}
		\end{equation}
		is negative when $\alpha<\pi /5$ and $u$ goes to infinity, so $|e^{i(1-2\gamma)z^2/2\xi}|$ is bounded. Finally, due to lemma \eqref{lemma:ppt}, $e^{-2\{\chi(z/a+z_0)-\chi(z_0)\}}$ is bounded.
		
		Now we have
		\begin{equation}
		\begin{split}
		|e^{i\gamma z^2/2}(\bar{R}(z/a+z_0)-\bar{R}(z_0\pm))|&\leq e^{\Re(i\gamma z^2/2)}\|\bar{R}'\|_{\infty}|z/a|\\
		&\leq c(tz_0^3)^{-1/2}\\
		\end{split}
		\end{equation}
		and
		\begin{equation}
		\begin{split}
		&\left|e^{i\gamma z^2/2}\left((\frac{2z_0}{z/a+2z_0})^{-2i\nu}-1\right)\right|\\
		&=\left|e^{i\gamma z^2/2}\int_{1}^{1+z/(2az_0)}(2i\nu)u^{2i\nu-1}du\right|\\
		&\leq |e^{i\gamma z^2/2}||z/(2az_0)|\sup_{u=1+\frac{sz}{2az_0},0\leq s\leq 1}|u^{2i\nu-1}|
		\end{split}
		\end{equation}
		Since 
		\begin{equation}
		\begin{split}
		|u^{2i\nu-1}|&=|e^{(2i\nu-1)(\log|u|+i \arg(u))}|\\
		&=e^{-\log|u|}e^{-2\nu \arg(u)}\\
		&=\frac{e^{-2\nu \arg(u)}}{|u|}
		\end{split}
		\end{equation}
		And it is easy to check that $|u|\geq \sin(\alpha)$ and since $\arg(u)$ is bounded, so $$ \sup_{u=1+\frac{sz}{2az_0},0\leq s\leq 1}|u^{2i\nu-1}|$$ is uniformly bounded with respect to $x,t$ and thus
		\begin{equation}
		\left|e^{i\gamma z^2/2}\left((\frac{2z_0}{z/a+2z_0})^{-2i\nu}-1\right)\right|\leq C(tz_0^3)^{-1/2}.
		\end{equation}
		Next write
		\begin{equation}
		\begin{split}
		&\left|e^{i\gamma z^2/2}\left(e^{-2(\chi(z/a+z_0)-\chi(z_0))}-1\right)\right|\\
		&\leq \sup_{0\leq s\leq 1}|e^{-2s(\chi(z/a+z_0)-\chi(z_0))}||2e^{i\gamma z^2/2}(\chi(z/a+z_0)-\chi(z_0))|\\
		&\leq C|e^{i\gamma z^2/2}||z/a||\log|z/a||\\
		&\leq C\frac{\log(tz_0^3)}{(tz_0^3)^{1/2}}
		\end{split}
		\end{equation}
		And finally we have
		\begin{equation}
		\begin{split}
		&|e^{i\gamma z^2/2}z^{-2\nu i}(e^{i(1-2\gamma)(z^2/2)\xi}-e^{i(1-2\gamma)z^2/2})|\\
		&\leq c|e^{i\gamma z^2/2}||z/a|\sup_{0\leq s\leq 1}|\frac{d}{ds}e^{i(1-2\gamma)z^2/2\xi(z;s)}|\\
		&\leq C(tz_0^3)^{-1/2}
		\end{split}
		\end{equation}
		Thus combining above estimates yields the expecting lemma. The rapid decay of $C(z_0)$ comes from the decomposition lemma.
	\end{proof}
	\begin{remark}
		By similar analysis, we have
		\begin{equation}
		\begin{split}
		\left|R(z/a+z_0)(\delta_B^1)^{2}-R({z_0}{\pm})z^{2\nu i}e^{-iz^2/2}\right|\\
		\leq C(z_0)e^{-\gamma \Im{z^2/2}}\left(t^{-1/2}+t^{-1/2}\log{(t)}\right)
		\end{split}
		\end{equation}
		on $L_B$.

	\end{remark}
	
	Moreover, on the contour $\Sigma_A$, there are similar estimates for $\bar{L}_A$ and $L_A$ too, which are
	\begin{equation}
	(N_A\delta e^{-it\theta})
	=\delta_A^0\delta_A^1	\end{equation}
	where

	\begin{equation}
	\begin{split}
	\delta_A^0(z)&= e^{\chi(-z_0)}e^{-it\theta(-z_0)}(2az_0)^{i\nu}\\
	\delta_A^1(z)&=
	e^{iz^2/4(1-z/(az_0)+z^2/(2a^2z_0^2)-z^3/(10a^3z_0^3))}\\
	&e^{\chi(z/a-z_0)-\chi(-z_0)}(-z)^{-i\nu}\left(\frac{-2z_0}{z/a-2z_0}\right)^{-i\nu}\\
	\end{split}
	\end{equation}
	And the analogs of lemma\eqref{lemma:phasereduction} are 
	\begin{equation}\label{eq:lemma for A 1}
	\begin{split}
	\left|\bar{R}(z/a-z_0)(\delta_A^1)^{-2}-\bar{R}((-z_0){\pm})(-z)^{2\nu i}e^{-iz^2/2}\right|\\
	\leq C(z_0)e^{-\gamma \Im{z^2/2}}\left(t^{-1/2}+t^{-1/2}\log{(t)}\right)
	\end{split}
	\end{equation}
	for $z\in \bar{L}_A$,and
	\begin{equation}\label{eq:lemma for A 2}
	\begin{split}
	\left|R(z/a-z_0)(\delta_A^1)^{2}-R((-z_0){\pm})(-z)^{-2\nu i}e^{iz^2/2}\right|\\
	\leq C(z_0)e^{-\gamma \Im{z^2/2}}\left(t^{-1/2}+t^{-1/2}\log{(t)}\right)
	\end{split}
	\end{equation}
	for $z\in L_A$.
	
	Then follow the same analysis  in DZ93\cite{deift_steepest_1993} we arrive at the following proposition:
	\begin{prop}
		\label{prop:split contributions}
		\begin{equation}
		\begin{split}
		q(x,t)&=\left[-2\int_{\Sigma_{A'}}((1-C_{w^{A'}})^{-1}I)w^{A'}(s)\frac{ds}{\pi}\right]_{12}\\
		&+\left[-2\int_{\Sigma_{B'}}((1-C_{w^{B'}})^{-1}I)w^{B'}(s)\frac{ds}{\pi}\right]_{12}\\
		&+O(t^{-k})+O(\frac{c(z_0)}{t}),\forall k\in \N
		\end{split}
		\end{equation}
		as $t\rightarrow \infty$
	\end{prop} 
	
	\section{Model RHP}
	In this section, we will transform the argumented RHP to a RHP on the real with a jump does not depend on $z$. Then by Louville's argument, we can solve the RHP explicitly and represented by solutions of the parabolic-cylinder equation. First we introduce some new notations following Deift-Zhou's method.
	Let $\hat{A}'=C_{\hat{w}^{A'}}:L^2(\hat{\Sigma}_{A'})\rightarrow L^2(\hat{\Sigma}_{A'})$ and let $\tilde{\Delta}_A^0:L^2(\hat{\Sigma}_{A'})\rightarrow L^2(\hat{\Sigma}_{A'})$ as the right multiple by $(\delta_A^0)^{\sigma_3}$. Then after shifting and rescaling, the new operator denotes $A:=C_{w^A}:L^2(\Sigma_{A})\rightarrow L^2(\Sigma_{A})$, where $w^A=(\Delta_A^0)^{-1}(N_A\hat{w}^{A'})\Delta_A^0$, and it has the relation with $\hat{A}'$:
	\begin{equation}
	\hat{A}'=N_A^{-1}(\tilde{\Delta}_A^0)^{-1}A\tilde{\Delta}_A^0N_A
	\end{equation}
	On the contour $\Sigma_A$, see figure \eqref{fig:sigma_A}, we have the RHP data for $A$ :$w^A=w^A_++w^A_-$, where $w^A_+=\begin{pmatrix}
	0 & (N_AR)(\delta_A^1)^2\\
	0& 0\\
	\end{pmatrix}$ and $w^A_-=\begin{pmatrix}
	0 & 0\\
	-(N_A\bar{R})(\delta_A^1)^2&0\\
	\end{pmatrix}$. Then base on the lemma \eqref{eq:lemma for A 1} and lemma \eqref{eq:lemma for A 2}, we have the RHP data for $A^0$. Set $v^{A^0}=(b_-^{A^0})^{-1}b_+^{A^0}=(I-w_-^{A^0})^{-1}(I+w_+^{A^0})$, where we define $w^{A^0}$ according to \eqref{eq:lemma for A 1} and \eqref{eq:lemma for A 2}, as
	\begin{eqnarray}
	w^{A^0}&=w_+^{A^0}=\begin{pmatrix}
	0 & R((-z_0)+)(-z)^{-2\nu i}e^{iz^2/2}\\
	0 & 0\\
	\end{pmatrix}\chi_{\{z\in \Sigma_A^2\}}\\
	&+\begin{pmatrix}
	0 & R((-z_0)-)(-z)^{-2\nu i}e^{iz^2/2}\\
	0 & 0\\
	\end{pmatrix}\chi_{\{z\in \Sigma_A^4\}}\\
	&=w^{A^0}_-=\begin{pmatrix}
	0 & 0\\
	-\bar{R}((-z_0)+)(-z)^{2\nu i}e^{-iz^2/2}&0\\
	\end{pmatrix}\chi_{\{z\in \Sigma_A^1\}}\\
	&+\begin{pmatrix}
	0 & 0\\
	-\bar{R}((-z_0)-)(-z)^{2\nu i}e^{-iz^2/2}&0\\
	\end{pmatrix}\chi_{\{z\in \Sigma_A^3\}}
	\end{eqnarray}
	where
	\begin{eqnarray}
	R((-z_0)+)&=\lim_{z\rightarrow (-z_0)+}\rho(z)=\frac{r(-z_0)}{1-|r(-z_0)|^2}\\
	R((-z_0)-)&=\lim_{z\rightarrow (-z_0)-}\rho(z)=-r(-z_0)
	\end{eqnarray}
	\begin{figure}[h]
		\centering
		\begin{tikzpicture}[scale=0.6]% coordinates
		\draw [<->](-2,-2)--(2,2);
		\draw (-3,-3)--(3,3);
		\draw [<->](-2,2)--(2,-2);
		\draw (-3,3)--(3,-3);
		\node [below] at (0,0) {0};
		\node  [right] at (3,3) {$\Sigma_A^2$};
		\node  [right] at (3,-3) {$\Sigma_A^1$};
		\node  [left] at (-3,3) {$\Sigma_A^3$};
		\node  [left] at (-3,-3) {$\Sigma_A^4$};
		\end{tikzpicture}
		
		\caption{Oriented contour $\Sigma_A$}
		\label{fig:sigma_A}
	\end{figure}
	Next we will show how to approximate the RHP data $w^{A'}$ by the data $w^{A^0}$. In fact, applying the restriction lemma and by changing variables, we can show that
	\begin{equation}
	\begin{split}
	\int_{\Sigma_{A'}}&((1-C_{w^{A'}})^{-1}I)w^{A'}(\xi)d\xi =\int_{\Sigma_{\hat{A}'}}((1-C_{\hat{w}^{A'}})^{-1}I)\hat{w}^{A'}(\xi)d\xi\\
	&=\int_{\Sigma_{\hat{A}'}}(N_A^{-1}(\tilde{\Delta}_A^0)^{-1}(1-A)^{-1}\tilde{\Delta}_A^0 N_AI)(\xi)\hat{w}^{A'}(\xi)d\xi\\
	&=\frac{1}{a}\int_{\Sigma_A}((1-A)^{-1}\Delta_A^0)(\xi)(\Delta_A^0)^{-1}(N_A\hat{w}^{A'})(\xi)\Delta_A^0(\Delta_A^0)^{-1}d\xi\\
	&=\frac{1}{a}\Delta_A^0\int_{\Sigma_A}((1-A)^{-1}I)w^A(\xi)d\xi (\Delta_A^0)^{-1}\\
	%& (\text{base on the lemma \eqref{eq:lemma for A 1} and lemma \eqref{eq:lemma for A 2},)\\
	&=\frac{1}{a}\Delta_A^0\int_{\Sigma_A}((1-A^0)^{-1}I)w^{A^0}(\xi)d\xi (\Delta_A^0)^{-1}+\frac{1}{a}O(t^{-1/2}+\frac{\log(t)}{t^{1/2}})
	\end{split}
	\end{equation}
	And combining with Proposition \eqref{prop:split contributions}, we have 
	\begin{prop}
		\label{prop:reduce to model rhp}
		\begin{equation}
		\begin{split}
		q(x,t)=&[\frac{-2}{a}\Delta_A^0\int_{\Sigma_A}((1-A^0)^{-1}I)w^{A^0}(\xi)\frac{d\xi}{\pi} (\Delta_A^0)^{-1}]_{12}\\
		&[\frac{-2}{a}\Delta_B^0\int_{\Sigma_B}((1-B^0)^{-1}I)w^{B^0}(\xi)\frac{d\xi}{\pi} (\Delta_B^0)^{-1}]_{12}\\
		&+O(t^{-k})+O(t^{-1}+\frac{\log(t)}{t}),\forall k\in \N,k>2
		\end{split}
		\end{equation}
		as $t\rightarrow \infty$
	\end{prop}
	Note $\int_{\Sigma_A}((1-A^0)^{-1}I)w^{A^0}(\xi)d\xi$ is connect to the following RHP, let \begin{equation}
	m^{A^0}(z)=I+\int_{\Sigma_A}\frac{((1-A^0)^{-1}I)w^{A^0}(\xi)}{\xi-z}\frac{d\xi}{2\pi i},\quad z\in \C\backslash \Sigma_A
	\end{equation}
	Then the corresponding RHP reads
	\begin{equation}
	\begin{cases}
	m^{A^0}_+(z)=m^{A^0}_-(z)v^{A^0}(z),\quad z\in \Sigma_A\\
	m^{A^0}(\infty)=I
	\end{cases}
	\end{equation}
	where 
	\begin{equation}
	v^{A^0}(z)=(1-w^{A^0}_-)^{-1}(1+w^{A^0}_+)
	\end{equation}
	Also we obtain that
	\begin{equation}
	m^{A^0}_1:=-Res(m^{A^0}(z),\infty)=\int_{\Sigma_A}((1-A^0)^{-1}I)w^{A^0}(\xi)\frac{d\xi}{2\pi i}.
	\end{equation}
	
	Similarly, we can compute for $\Sigma_B$, and since the reflection coefficient has the property that $r(z)=-\bar{r}(-\bar{z})$ and note that all the jump matrices are triangle matrix, we have the following relation:
	\begin{equation}
	\sigma_3\overline{v^{B^0}(-\bar{z})}\sigma_3=v^{A^0}(z)
	\end{equation}
	Moreover, by uniqueness of the RHP, 
	\begin{equation}
	m^{A^0}(z)=\sigma_3\overline{m^{B^0}(-\bar{z})}\sigma_3
	\end{equation}
	which implies that
	\begin{equation}
	m^{B^0}_1=-\sigma_3\overline{m^{A^0}_1}\sigma_3
	\end{equation}
	Now form the Proposition\eqref{prop:reduce to model rhp}, it follows
	\begin{equation}
	\label{eq:finall asymptotics}
	q(x,t)=\frac{-2}{a}\left[(\delta_A^0)^2(m^{A^0}_1)_{12}+\overline{(\delta_A^0)^{2}(m^{A^0}_1)_{12}}\right]+O(\frac{\log(t)}{t})
	\end{equation}
	as $t\rightarrow \infty$.
	In the rest of the section, we will solve the model RHP in terms of solutions of the parabolic-cylinder equation. The basic idea is to "close the lens", which is the inverse processing of the contour deformation ("open lens").
	\begin{figure}[h]
		\centering
		\begin{tikzpicture}[scale=0.6]% coordinates
		\draw [<-<](-2,-2)--(2,2);
		\draw (-3,-3)--(3,3);
		\draw [<-<](-2,2)--(2,-2);
		%	\draw [-<](0,0)--(2,2);
		\draw (-3,3)--(3,-3);
		\node [below] at (0,0) {0};
		\draw [<-<] (-2,0)--(2,0);
		\draw [-](-4,0)--(4,0);
		\node [right] at (2,1) {$\Omega_1^e$};
		\node [] at (0,2) {$\Omega_2^e$};
		\node [left] at (-2,1) {$\Omega_3^e$};
		\node [left] at (-2,-1) {$\Omega_4^e$};
		\node [] at (0,-2) {$\Omega_5^e$};
		\node [right] at (2,-1) {$\Omega_6^e$};
		%	\node [right] at (2,1) {$\Omega_1$};
		%	\node  [right] at (3,3) {$\Sigma_A^2$};
		%	\node  [right] at (3,-3) {$\Sigma_A^1$};
		%	\node  [left] at (-3,3) {$\Sigma_A^3$};
		%	\node  [left] at (-3,-3) {$\Sigma_A^4$};
		\end{tikzpicture}
		
		\caption{Oriented contour $\Sigma_A$}
		\label{fig:sigma_e}
	\end{figure}
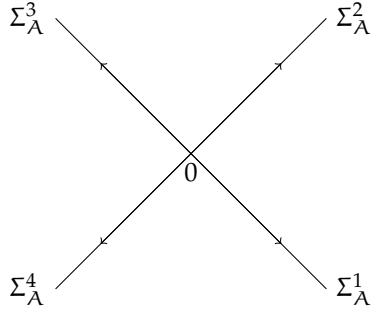
	
	First we reorient the right-half of $\Sigma_A$ , denote the new contour as $\Sigma_{A,r}$, meanwhile the new RHP data on the right half plane become $w^{A,r}_{\pm}=-w^{A^0}_{\mp}$, then  extend the contour $\Sigma_{A,r}$ to $\Sigma_e=\Sigma_{A,r}\cup \R$ by assigning 0 to the RHP data on $\R$ and mark the six regions as shown on the Fig, then define a matrix $\phi$ as
	\begin{equation}
	\phi(z)=(-z)^{vi\sigma_3}\times
	\begin{cases}
	1 & z\in \Omega_2^e\cup \Omega_5^e\\
	(b_+^{A^0})^{-1} & z\in \Omega_1^e\cup \Omega_4^e\\
	(b_-^{A^0})^{-1} & z\in \Omega_3^e\cup \Omega_6^e\\
	\end{cases}
	\end{equation}
	Conjugating $v^{A^0}$ by $\phi_-(z)v^{A^0}\phi^{-1}_+(z)$, denotes as $v^{A^0,\phi}$, we have a new RHP which only has jumps on the real line. And the jump on the real line is
	\begin{equation}
	\begin{split}
	v^{A^0,\phi}&=\phi_-(z)\phi^{-1}_+(z)\\
	&=(-z)^{i\nu \sigma_3}_-((b_-^{A^0})^{-1})b_+^{A^0}(-z)^{-i\nu \sigma_3}_+\\
	&=e^{iz^2/4\hat{\sigma}_3}\begin{pmatrix}
	1 & 0\\
	-\frac{\bar{r}(-z_0)}{1-|r(-z_0)|^2}&1\\
	\end{pmatrix}(-z)^{i\nu \sigma_3}_-(-z)^{-i\nu \sigma_3}_+
	\begin{pmatrix}
	1 & \frac{r(-z_0)}{1-|r(-z_0)|^2}\\
	0  & 1
	\end{pmatrix}\\
	&=e^{iz^2/4\hat{\sigma}_3}\begin{pmatrix}
	1 & 0\\
	-\frac{\bar{r}(-z_0)}{1-|r(-z_0)|^2}&1\\
	\end{pmatrix}(1-|r(-z_0)|^2)^{\sigma_3}
	\begin{pmatrix}
	1 & \frac{r(-z_0)}{1-|r(-z_0)|^2}\\
	0  & 1
	\end{pmatrix}\\
	&=e^{iz^2/4\hat{\sigma}_3}\begin{pmatrix}
	1-|r(-z_0)|^2 & r(-z_0)\\
	-\bar{r}(-z_0)&1\\
	\end{pmatrix}\\
	&=e^{iz^2/4\hat{\sigma}_3}v(-z_0)
	\end{split}
	\end{equation}
	
	Then let $H(z)=m^{A^0}(z)\phi^{-1}(z)$, is satisfies the following RHP :
	\begin{equation}
	\begin{cases}
	H_+(z)=H_-(z)e^{iz^2/4\hat{\sigma}_3}v(-z_0), \quad z\in \R\\
	H(\infty)=(-z)^{\nu i \sigma_3}
	\end{cases}
	\end{equation}
	Let $\Psi=He^{iz^2/4\sigma_3}$, then $\Psi_+=\Psi_-v(-z_0)$, which has a constant jump over the real line. Then it is easy to check that $\frac{d\Psi}{dz}\Psi^{-1}$ has no jump on the real line hence is entire then by Liouville's argument, we have
	\begin{equation}
	\begin{split}
	\frac{d\Psi}{dz}\Psi^{-1}&=\frac{dH}{dz}H^{-1}+H\sigma_3H^{-1}\frac{iz}{2}\\
	&=\frac{iz}{2}\sigma_3+\frac{i}{2}[\sigma_3,m^{A^0}_1]+O(z^{-1})\\
	&\equiv\frac{iz}{2}\sigma_3+\frac{i}{2}[\sigma_3,m^{A^0}_1]
	\end{split}
	\end{equation}
	Let $\beta=\frac{i}{2}[\sigma_3,m^{A^0}_1]=\begin{pmatrix}
	0 & \beta_{12}\\
	\beta_{21} & 0\\
	\end{pmatrix}$, it follows that 
	\begin{equation}
	\label{eq:pre weber's parabolic cylinder equation}
	\frac{d\Psi}{dz}=(\frac{iz}{2}\sigma_3+\beta)\Psi.
	\end{equation}
	First consider $\Im z>0$, from the equation \eqref{eq:pre weber's parabolic cylinder equation}, we obtain two second order ODEs:
	\begin{eqnarray}
	\frac{d^2}{dz^2}\Psi_{11}^{+}&=(i/2-z^2/4+\beta_{12}\beta_{21})\Psi_{11}^+\\
	\frac{d^2}{dz^2}\Psi_{21}^{-}&=(-i/2-z^2/4+\beta_{12}\beta_{21})\Psi_{21}^-\\
	\end{eqnarray}
	By setting $\Psi_{11}^+(z)=g(e^{-3\pi i/4}z)$, we have
	\begin{equation}
	\label{eq: standard equation}
	\frac{d^2}{dz^2}g(z)-(\frac{z^2}{4}+a)g(z)=0
	\end{equation}
	where $a=-\frac{1}{2}+i\beta_{12}\beta_{21}.$ This is the Weber's parabolic cylinder equation, search this on the Digital Library of Mathematics Functions, we have the asymotics for the solutions when $z\rightarrow \infty$, for reader's convenience, we copy the asymptotic expansions here:
	\begin{equation}
	\begin{split}
	U(a,z)&\sim e^{-\frac{1}{4}z^2}z^{-a-1/2}\sum_{s=0}^{\infty}(-1)^s\frac{(1/2+a)_s}{s!(2z^2)^s},\quad |\arg(z)|<\frac{3\pi }{4}\\
	&\sim e^{-\frac{1}{4}z^2}z^{-a-1/2}\sum_{s=0}^{\infty}(-1)^s\frac{(1/2+a)_s}{s!(2z^2)^s}\\
	&\pm i\frac{\sqrt{2\pi}}{\Gamma(1/2+a)}e^{\mp i\pi a}e^{\frac{1}{4}z^2}z^{a-1/2}\sum_{s=0}^{\infty}(-1)^s\frac{(1/2-a)_s}{s!(2z^2)^s},\quad \frac{1}{4}\pi<\pm \arg(z)<\frac{5}{4}\pi
	\end{split}
	\end{equation}
	from the digital library, we know that the Wronskian $W\{U(a,z),U(a,-z)\}=\frac{\sqrt{2\pi}}{\Gamma(1/2+a)}$ is non-zero as long as $a+1/2$ is not a non-positive integer. 
	For now, assume that it is true, then the solution of the equation\eqref{eq: standard equation} can be represented by
	\begin{equation}
	g(z)=c_1U(a,z)+c_2U(a,-z).
	\end{equation}
	And we know that as $z=e^{1/4\pi i}\sigma \rightarrow \infty$, $\Psi_{11}^{+}=(-e^{i\pi/4})^{i\nu}e^{-\sigma^2/4}=e^{\nu i (\log\sigma-i\frac{3}{4}\pi)}e^{-\sigma^2/4}$, compare it with the asymptotic expansion of $g$, we have 
	\begin{equation}
	c_2=0,a=-\nu i -1/2,c_1=e^{\frac{3}{4}\pi\nu}
	\end{equation}
	so that
	\begin{equation}
	\Psi_{11}^{+}(z)=e^{\frac{3}{4}\pi\nu}U(a,e^{-\frac{\pi}{4}i}z),\quad \Im z>0
	\end{equation}
	Similary, we have for $\Im z<0$,
	\begin{equation}
	\Psi^-_{11}(z)=e^{-\frac{\pi\nu}{4}}U(a,e^{\frac{3\pi i}{4}}z)
	\end{equation}
	Meanwhile  we have
	$\Psi_{21}=\beta_{12}^{-1}\left(\frac{d}{dz}\Psi_{11}-\frac{iz}{2}\Psi_{11}\right)$, so $\Psi_{21}^{\pm}$ can be automatically represented by $\Psi_{11}^{\pm}$. Also we have
	$$\Psi_-^{-1}\Psi_+=v(-z_0)=\begin{pmatrix}
	1-|r(-z_0)|^2& r(-z_0)\\
	-\bar{r}(-z_0)&1\\
	\end{pmatrix},$$ comparing both sides we have the following relation:
	\begin{equation}
	\begin{split}
	-\bar{r}(-z_0)&=\Psi_{11}^-\Psi_{21}^+-\Psi_{21}^-\Psi_{11}^{+}\\
	&=\beta_{12}^{-1}[\Psi_{11}^-(\Psi_{11}^+)'-(\Psi_{11}^-)'\Psi_{11}^+]\\
	&=\beta_{12}^{-1}e^{\pi \nu/2}W\{U(a,e^{3\pi i/4}z),U(a,e^{-\pi i/4}z)\}\\
	&=\frac{e^{\pi \nu/2}e^{3\pi i/4}\sqrt{2\pi}}{\beta_{12}\Gamma(-\nu i)}\quad\quad (\text{see [DLMF] equation (12.2.11)})
	\end{split}
	\end{equation}
	
	Thus,
	\begin{equation}
	\beta_{12}=\frac{e^{\pi \nu /2}\sqrt{2\pi}e^{3\pi i/4}}{-\bar{r}(-z_0)\Gamma(-\nu i)}
	\end{equation}
	and
	\begin{equation}
	\beta_{21}=-\nu/\beta_{12}=\frac{e^{\pi \nu /2}\sqrt{2\pi}e^{-3\pi i/4}}{r(-z_0)\Gamma(\nu i)}.
	\end{equation}
	
	As mentioned before, we assume the Wroskian is non-zero. In fact it is true provided that $\frac{1}{\Gamma(1/2+a)}=\frac{1}{\Gamma(-iv)}$ is not zero since $\nu=-\frac{1}{2\pi}\log(1-|r(-z_0)|^2)>0$.
	Note also we have
	\begin{eqnarray}
	(m_1^{A^0})_{21}&=&i\beta_{21}\\
	(m_1^{A^0})_{12}&=&-i\beta_{12}
	\end{eqnarray}
	Finally, substituting back to equation \eqref{eq:finall asymptotics}, we obtain
	\begin{equation}
	\begin{split}
	q(x,t)&=\frac{-2}{a}\left[(\delta_A^0)^2(m^{A^0}_1)_{12}+\overline{(\delta_A^0)^{2}(m^{A^0}_1)_{12}}\right]+O(\frac{\log(t)}{t})\\
	&=\frac{-2}{a}[e^{2\chi(-z_0)}e^{-2it\theta(-z_0)}(2az_0)^{2i\nu}\frac{e^{\pi \nu /2}\sqrt{2\pi}e^{5\pi i/4}}{\bar{r}(-z_0)\Gamma(-\nu i)}\\
	&+e^{-2\chi(-z_0)}e^{2it\theta(-z_0)}(2az_0)^{-2i\nu}\frac{e^{\pi \nu /2}\sqrt{2\pi}e^{3\pi i/4}}{r(-z_0)\Gamma(\nu i)}]
	\\&+O(\frac{\log(t)}{t})\\
	\end{split}
	\end{equation}
	as $t\rightarrow \infty.$ By simplifying this we get the result\eqref{main result} as being reported at the introduction .
	%%%%%%%%%%%%%%%%%%%%%%%%%%%%%%%%%%%%%%%%%%%%%%%%%%%%%%%%%%%%%%%%%%%%%
	%
	% To add references to your document, replace the two \bib commands below.
	%
	%         1. You can use a list of \bib commands for the items you reference as is
	%         done in our toy example here.
	%
	%         2. A second option is to use the command
	%             \bibselect{yourltbfile}
	%         to point to a file of \bib commands that should be named
	%         yourltbfile.ltb and be placed in the same folder as your LaTeX
	%         source files.
	%
	%         3. A third option is to use the command
	%             \bibliography{yourbibfile}
	%         to point to a file of BibTeX \bib commands that should be named
	%         yourltbfile.bbl and be placed in the same folder as your LaTeX
	%         source files.
	%
	% If you use option 3. above, you should comment out or delete the lines
	%            \begin{bibdiv}
	%                \begin{biblist}
	%        before the \bib command below as well as the line
	%                  \end{biblist}
	%              \end{bibdiv}
	%        after it.
	%
	% If you use options 2. or 3. and wish to make your source file self-contained you may
	%         for final submission, simply copy the \bib entries to your \LaTeX\ file and
	%         wrap them, if necessary, as indicated above.
	%
	%%%%%%%%%%%%%%%%%%%%%%%%%%%%%%%%%%%%%%%%%%%%%%%%%%%%%%%%%%%%%%%%%%%%%

\bibliography{testbib}

\end{document}